\documentclass[review]{elsarticle}

\usepackage{lineno,hyperref}
\modulolinenumbers[5]

\usepackage{lineno}
\usepackage{latexsym,amssymb,amsmath}
\usepackage{mathabx}
\usepackage{longtable}
\usepackage{tikz,pgf}
\usepackage{subfig}
\usepackage{wrapfig}
\usepackage{rotating}
\usepackage{multirow}
\usepackage{hyperref}
\usepackage{umoline}
\usepackage{enumitem}
\usepackage{booktabs}

\usepackage{amsthm}
\theoremstyle{plain}
\newtheorem{theorem}{Theorem}[section]

\theoremstyle{definition}

\theoremstyle{remark}

\usepackage{algorithm}
\usepackage{algorithmic}

\usepackage{graphicx}
\usepackage{epstopdf}
\input{def.set}

\usepackage{url}
\usepackage{hyperref}
\hypersetup{colorlinks,%
citecolor=blue,%
filecolor=blue,%
linkcolor=red,%
urlcolor=blue,%
pdftex}

\journal{Arxiv}

\begin{document}

\begin{frontmatter}

\title{Harmonic holes as the submodules of brain network and network dissimilarity}

\author[a]{Hyekyoung Lee}
\author[d]{Moo K. Chung} 
\author[a]{Hongyoon Choi}
\author[b]{Hyejin Kang}
\author[a]{Seunggyun Ha} 
\author[c]{Yu Kyeong Kim} 
\author[a,b]{Dong Soo Lee}

\address[a]{Seoul National University Hospital,} 
\address[b]{Seoul National University,}
\address[c]{SMG-SNU Boramae Medical Center, Seoul, Republic of Korea}
\address[d]{University of Wisconsin, Madison, WI, USA}

%
%
%

\begin{abstract} 
Persistent homology has been applied to brain network analysis for finding the shape of brain networks across multiple thresholds. 
In the persistent homology, the shape of networks is often 
quantified by  the sequence of $k$-dimensional holes and 
Betti numbers.
The Betti numbers are more widely used than holes themselves in topological brain network 
analysis. 
However, the holes show the local connectivity of networks, and they can be very informative features in analysis. 
In this study, we propose a new method 
of measuring 
network differences based on the dissimilarity measure of harmonic holes (HHs).
The HHs, which represent the substructure of brain networks, are extracted by the Hodge Laplacian of brain networks.
We also find the most contributed HHs to the network difference based on the HH dissimilarity.
We applied our proposed method to 
clustering the networks 
of 4 groups, normal control (NC), stable and progressive mild cognitive impairment (sMCI and pMCI), and Alzheimer's disease (AD). 
The results showed that 
the clustering performance of the proposed method was better than that of network distances based on only the global change of topology. 
\end{abstract} 

\end{frontmatter}

\section{Introduction}

Persistent homology has been widely applied to brain network analysis for finding the topology of networks in multiscale \cite{chung.2009.ipmi,lee.2012.ieeemi,singh.2008.jv,solo.2018.tmi} Since a `simplicial complex' is not a familiar term in brain network analysis, we refer to it as a `network' that is generally used. 
It quantifies the shape of brain networks by using $k$-dimensional holes and their cardinality, 
the $k$th Betti number \cite{carlsson.2005.ijsm,edelsbrunner.2008.cm}. 
A persistence diagram (PD) summarizes the change of Betti numbers during the filtration of networks by recording when and how 
holes appear and disappear during the filtration. 
The persistent homology also provides distances for distinguishing networks such as the bottleneck distance and kernel-based distances \cite{edelsbrunner.2008.cm,reininghaus.2015.cvpr}. 
Such distances mostly find network differences in their PDs.
The Betti numbers and PDs are more often 
used than holes themselves in network applications. 

Holes represent the submodule of brain networks.   
$0$-dimensional holes, i.e., connected components, modules or clusters have been widely studied for finding functional or structural submodules in a brain \cite{chung.2017.ipmi,lee.2012.ieeemi,sporns.2016.arp}. 
On the other hand, $1$-dimensional holes have been rarely used for brain network analysis \cite{chung.2018.arxiv,lee.2014.miccai,lee.2018.isbi,petri.2014.jrsi,sizemore.2018.jcn,sporns.2000.cc}. 
Most studies in brain network analysis do not use $2$- and higher order simplexes in networks since networks. Therefore, all cycles in a network are considered as $1$-dimensional holes.  
There are few network measures based on cycles in brain network analysis such as cycle probability and the change of the number of cycles during graph filtration \cite{chung.2018.arxiv,sporns.2000.cc}. 
These measures helped to compare the global property of networks 
but could not find the discriminative substructures of networks.

If 
higher order simplexes are introduced in a network, the number of $1$-dimensional holes is 
significantly reduced 
due to the removal of  filled-in triangles. 
The previous brain network studies that studied 
higher order simplexes mostly found holes based on Zomorodian and Carlsson's (ZC) algorithm \cite{petri.2014.jrsi,sizemore.2018.jcn,zomorodian.2005.dcg}. 
The ZC algorithm is very fast in linear-time, however, it finds the sparse representation of a hole that identifies only one path 
 around the hole and ignores the other paths.
This introduces an ambiguity in hole identification in practice. 
A better approach would be to localize the holes 
%
by the eigen-decomposition of Hodge Laplacian of a network. Such holes are called as the harmonic holes (HHs). 
The HH shows all possible paths around the hole with their weights \cite{friedman.1996.astc,horak.2013.am,lim.2015.arxiv}. 
The HHs have been applied to brain network analysis for localizing persistent holes \cite{lee.2018.isbi,lee.2014.miccai}. 
The $1$-dimensional holes in a network with higher order simplexes 
have at least one indirect path between every two nodes. 
Thus, the holes are related to the abnormality or inefficiency of the network. 
The previous studies found the persistent holes with long duration in a network as abnormal holes, and localized them by harmonic holes. 
Therefore, the duration of holes was used instead of HHs in network discrimination. 

In this paper, we propose a new 
measure for estimating network dissimilarity based on persistent HHs (HH dissimilarity). 
The proposed HH dissimilarity is motivated from the bottleneck distance. 
The bottleneck distance first estimates the correspondence between holes between networks 
that are represented by 
points in PDs, and then chooses the maximum among all the distances between the estimated pairs of holes \cite{cohen-steiner.2007.dcg}. 
The HH dissimilarity also estimates the correspondence between HHs of two different networks that are represented by real-valued eigenvectors, and takes the averaged dissimilarities of the estimated pairs of HHs. 
The advantage of HH dissimilarity is not only to measure the network differences but also to quantify a HH's  
contributions to the network differences. 
We  will call the amount  of HH's 
of contribution the {\em citation} of HH. 
This allows us to identify the discriminative subnetworks of networks. 

The proposed method is applied to metabolic brain networks obtained from 
the FDG PET dataset  in Alzheimer's disease neuroimaging initiative (ADNI). 
The dataset consists of 4 groups: 
normal controls (NC), stable and progressive mild cognitive impairment (sMCI and pMCI), and Alzheimer's disease (AD). 
We generated 2400 networks by bootstrap, 
and compared the clustering performance with the existing network distances such as L$_{2}$-norm (L2) of the difference between distance matrices, Gromov-Hausdorff (GH) distance, Kolmogorov-Smirnov (KS) distance of connected components and cycles (KS$_0$ and KS$_1$), and bottleneck distance of holes \cite{carlsson.2008.ijcv,chung.2017.ipmi,chung.2018.arxiv,cohen-steiner.2007.dcg,lee.2012.ieeemi}.  
The results showed that the HH dissimilarity had the superior clustering performance than the other distance measures, and comparing local connectivities could be more helpful to discriminating the progression of Alzheimer's disease.

\section{Materials and methods} 

\subsection{Data sets, preprocessing, and the construction of metabolic connectivity}
\label{sec:dataset}
 
We used FDG PET images in ADNI data set (http://adni.loni.usc.edu). 
The ADNI FDG-PET dataset consists of 4 groups: 181 NC, 91 sMCI, 77 pMCI, and 135 AD (Age: $73.7 \pm 5.9$, range $56.1 \sim 90.1$). 
FDG PET images were measured 30 to 60 minutes and they were averaged over all frames. 
The voxel size in the images were standardized in 
$1.5 \times 1.5 \times 1.5$ mm resolution.
The images were spatiallly normalizd to Montreal Neurological Institute (MNI) space using statistical parametric mapping (SPM8, www.fil.ion.ucl.ac.uk/spm). 
The details of data sets and preprocessing are given in \cite{choi.2018.bbr}. 
The whole brain image was parcellated into 
94 regions of interest (ROIs) based on automated anatomical labeling (AAL2) excluding  cerebellum  \cite{rolls.2015.ni}. 
The 94 ROIs served as network nodes and their measurements were obtained by averaging FDG uptakes in the ROI. The averaged FDG uptake was globally normalized by the sum of 94 averaged FDG uptakes. 
The distance between 2 nodes was estimated by the diffusion distance on positive correlation between the measurements. 
The diffusion distance considers an average distance of all direct and indirect paths between 2 nodes via random walks \cite{coifman.2005.pnas}. 
The diffusion distance is known to be more robust to noise and outliers.

\subsection{Harmonic holes} 

\subsubsection{Simplicial complex}
The algebraic topology extends the concept of a graph further to a simplicial complex. 
Suppose that a non-empty node set $V$ is given. 
If the set of all subsets of $V$ is denoted by $2^{V},$ an abstract simplicial complex $K$ is a subset of $2^{V}$ such that (1) $\emptyset \in K$, and (2) if $\sigma \in K$ and $\tau \in \sigma$, $\tau \in K$ \cite{edelsbrunner.2008.cm,edelsbrunner.2009.book}. 
Each $\sigma \in K$ is called a simplex. 
A $i$-dimensional simplex is an element with $i+1$ nodes, $v_{1},...,v_{i+1} \in V$, denoted by $\sigma_{i}=[v_{1},...,v_{i+1}]$.  
The dimension of $K$, denoted as $\mbox{dim} K$, is the maximum dimension of a simplex $\sigma \in K$.
The collection of $\sigma_i$'s in $K$ is denoted by $K_{i}$ $(-1 \le i \le \mbox{dim} K)$.  
The number of simplices in $K_{i}$ is denoted as 
$|K_{i}|$. 
The $i$-skeleton of $K$ is defined as 
$K^{(i)} = K_{0} \cup \cdots \cup K_{i}$ $(0 \le i \le \mbox{dim} K)$.  
Thus, a graph with nodes and edges is $1$-skeleton $K^{(1)}$.  
In this paper, we will only consider $2$-skeleton $K^{(2)}$ of a simplicial complex that includes nodes, edges, and triangles. 
For convenience, we call it \emph{a (simplicial) network} \cite{kim.2018.lma}. 

\subsubsection{Incidence matrix} 

We denote a $|K_{i}|$-dimensional integer space as $\Integer^{|K_{i}|}.$ 
Given a finite simplicial complex $K,$ 
       a chain complex $C_{i}$ is defined in $\Integer^{|K_{i}|}$   \cite{edelsbrunner.2008.cm,zomorodian.2005.dcg}.  
The boundary operator $\partial_{i}$ and coboundary operator $\partial_{i}^{\top}$ for $i = 1, \dots, N$ $(N>0)$ are functions such that $\partial_{i}: C_{i} \rightarrow C_{i-1}$ and $\partial_{i}^{\top}: C_{i-1} \rightarrow C_{i}$, respectively. 
We define $\partial_{i} = 0$ for $i < 1$ or $i > N$.

Given $\sigma_{i} = [v_{1},...,v_{i+1}] \in C_{i}$, the boundary of $\sigma_{i}$ is algebraically defined as 
$$ \partial_{i} \sigma_{i} = \sum_{j=1}^{i+1} (-1)^{j-1} [v_{1}, \dots, v_{j-1}, v_{j+1}, \dots, v_{i+1}].$$ 
If the sign of $\sigma_{i-1}$ in $\partial_{i} \sigma_{i}$ is positive/negative, it is called positively/negatively oriented with respect to $\sigma_{i}.$
We denote the positive/negative orientation by $\sigma_{i-1} \in_{+/-} \sigma_{i}$.   
The boundary of the boundary is always zero, i.e., $\partial_{i-1} \partial_{i} = 0$.

If the simplicial complex $K$ has 
$$K_{i} = \left\{ \sigma_{i}^{1}, \cdots, \sigma_{i}^{|K_{i}|} \right\},  \quad K_{i-1} = \left\{ \sigma_{i-1}^{1}, \cdots ,\sigma_{i-1}^{|K_{i-1}|} \right\},$$
 the boundary operator $\partial_{i}$ is represented by the $i$th incidence matrix $\bM_{i} \in \Integer^{|K_{i-1}| \times |K_{i}|}$ such that  \cite{friedman.1996.astc,horak.2013.am,lim.2015.arxiv} 
\be
[\bM_{i}]_{mn} = \left\{ \begin{array}{rl}
1 & \mbox{if } \sigma_{i-1}^{m} \in_{+} \sigma_{j}^{n}, \\
-1 & \mbox{if } \sigma_{i-1}^{m} \in_{-} \sigma_{j}^{n}, \\
0 & \mbox{otherwise.} 
\end{array} \right. 
\label{eq:incidence}
\ee
The coboundary operator $\partial_{i}^{\top}$ is represented by $\bM_{i}^{\top}$. 
$\sigma_{i}^{n}$ in $K_{i}$ is represented by a vector in $\Integer^{|K_{i}|}$ in which the $n$th entry is 1 and the rest is 0. 
The linear combination of $\sigma_{i}$'s can be represented by the linear combination of $|K_{i}|$-dimensional vectors.   

\subsubsection{Combinatorial Hodge Laplacian} 

A combinatorial Hodge Laplacian $\bL_{i}: C_{i} \rightarrow C_{i}$ is defined by 
\be 
\label{eq:laplacian}
\bL_{i} = \bL_{i}^{up} + \bL_{i}^{down} = \bM_{i+1} \bM_{i+1}^{\top} + \bM_{i}^{\top} \bM_{i},
\ee  
where $\bL_{i}^{up} \in \Integer^{|K_{i}| \times |K_{i}|}$ and $\bL_{i}^{down} \in \Integer^{|K_{i}| \times |K_{i}|}$ are composite functions $\partial_{i+1} \partial_{i+1}^{\top}: C_{i} \rightarrow C_{i+1} \rightarrow C_{i}$ and $\partial_{i}^{\top} \partial_{i}: C_{i} \rightarrow C_{i-1} \rightarrow C_{i}$, respectively \cite{friedman.1996.astc,horak.2013.am,kim.2018.lma,lim.2015.arxiv}
The kernel and image of $\bL_{i}$ are denoted by $\mbox{ker} \bL_{i}$ and $\mbox{img} \bL_{i},$ respectively. 
The $\mbox{ker} \bL_{i}$ is called harmonic classes $H_{i}$ \cite{kim.2018.lma}.

The $i$th homology and cohomology groups of $C = \left\{ C_{i}, \partial_{i} \right\}$ are defined respectively by 
$$\tilde{H}_{i} (C) = \mbox{ker} \partial_{i} / \mbox{img} \partial_{i+1} \mbox{ and }  \tilde{H}^{i} (C) = \mbox{ker} \partial_{i+1}^{\top} / \mbox{img} \partial_{i}^{\top}. $$

\begin{theorem}[Combinatorial Hodge Theory \cite{friedman.1996.astc,kim.2018.lma,lim.2015.arxiv}] 
Suppose that a chain complex $\left\{ C_{i}(X; \Real), \partial_{i} \right\}$ is given for $i = 0, \dots, N$, and $C_{i}(X; \Real)$ is considered as an $\Real$-vector space. 
Harmonic classes $H_{i}$ obtained by the combinatorial Laplacian $\bL_{i}$ are congruent to the $i$th homology and cohomology groups, $\tilde{H}_{i}$ and $\tilde{H}^{i}$ of $C$, i.e., 
$$H_{i} \cong \tilde{H}_{i} (C; \Real) \cong \tilde{H}^{i} (C; \Real).$$
\end{theorem}  
\begin{proof} 
$\mbox{rank} H_{i} = \mbox{rank} C_{i} - \mbox{rank} \bL_{i} = \mbox{rank} C_{i} - (\mbox{rank} \partial_{i} + \mbox{rank} \partial_{i+1}) = \mbox{rank}  \tilde{H}_{i} (C; \Real).$ 
\end{proof} 


The harmonic classes $H_{i}=\mbox{ker}  \bL_{k}$ is also called a harmonic space \cite{kim.2018.lma}. 
The homology group $\tilde{H}_{i}$ in persistent homology can be replaced with a harmonic space $H_{i}$, and the rank of $H_{i}$ is the same as the $i$th Betti number. 
We call a hole in $H_{i}$  a harmonic hole (HH), and a hole in $\tilde{H}_{i}$ estimated by Smith normal form 
a binary hole \cite{zomorodian.2005.dcg}.

Given a simplicial network with $p$ nodes, $q$ edges, and $r$ filled-in triangles, we estimate $\bL_{1} \in \Integer^{q \times q}$ in (\ref{eq:laplacian}), and $H_{i} = \left\{ \bx \in \Real^{q \times 1} | \bL_{1} \bx = 0 \right\}$. 
The eigenvector of $\bL_{1}$ with zero eigenvalue, $\bx \in \Real^{q \times 1}$ represents a HH. 
The entry of $\bx$ can be positive or negative depending on the orientation of edges in the hole. 
The absolute value of the entry of $\bx$ represents the weight of the corresponding edge in the hole.  
Since $\bx$ and $-\bx$ have zero eigenvalue, they represent the same hole, and $\parallel \bx \parallel = 1.$

\subsubsection{Computing persistent HHs}

In this study, we have the distances between pairs of nodes in a brain network. 
Given a set of nodes and their distances, Rips complex with threshold $\epsilon$ is the clique complex induced by a set of edges with their distances less than $\epsilon$. 
Rips filtration is the nested sequence of Rips complexes obtained by increasing threshold $\epsilon$. 
To compute persistent holes over threshold,  we perform Rips filtration on brain network nodes 
\cite{carlsson.2005.ijsm,edelsbrunner.2008.cm}. 

Zomorodian and Carlsson developed an efficient algorithm for computing persistent holes based on the Smith normal form 
\cite{zomorodian.2005.dcg}.  
It is an incremental algorithm that updates the range and null spaces of incidence matrices during Rips filtration. 
The representation of a persistent binary hole is changed by adding simplexes during Rips filtration. 
The ZC algorithm chose the youngest binary hole at the birth 
of the persistent hole. 
The ZC algorithm is fast in practically linear-time, 
however,
the obtained binary hole shows only one path around the  
hole and the other paths are ignored. 
On the other hand, a HH shows all possible paths around the persistent hole, and represents the contribution of a path to the generation of the hole by edge weights in the path.  
Thus, the HH is better in localizing 
a persistent hole than a binary hole when we want to extract local connectivity in a brain network. 
However, there is no algorithm for estimating persistent HHs during the filtration in literature.

In this study, we will estimate the youngest persistent HHs just like the ZC algorithm. 
First, we sort edges $e_{1}, \dots, e_{q}$ in the ascending order of an edge distance, and perform the Rips filtration by the fast ZC algorithm. 
To avoid having the same edge distance, we select the ordered index $1, \dots, q$ as the filtration value, instead of the edge distance. 
The reason for performing the ZC algorithm first is that the computation of eigen-decomposition at every filtration value is too expensive. 
Then, we obtain a PD which is the set of the birth and death thresholds of persistent holes.
If a persistent hole appears at $i_{X}$ and disappears at $i_{Z},$ we perform the eigen-decomposition of Hodge Laplacian at $i_{X}, i_{Z},$ and $i_{Y}=i_{Z}-1$ to estimate the corresponding HH. 
The $i_{Y}$ is the threshold just before the death of the persistent hole.

The harmonic spaces at $i_{X}, i_{Y},$ and $i_{Z}$ are written by matrices 
$$H_{X} = [\bx_{1}, \cdots, \bx_{l}] \in \Real^{q \times l}, H_{Y} = [\by_{1}, \cdots, \by_{m}] \in \Real^{q \times m},  H_{Z} = [\bz_{1}, \cdots, \bz_{n}] \in \Real^{q \times n},$$ respectively.
The HH appearing at $i_{X}$ and disappearing at $i_{Z}$ will be in $H_{X}$ and $H_{Y},$ but not in $H_{Z}.$   
We find which $\by \in H_{Y}$ does not depend on $\bz_{i}$'s in $H_{Z}.$
If $\by \in H_{Y}$ depends on $H_{Z},$ the smallest singular value of the matrix $[H_{Z}, \by]$ is close to $0.$  It implies 
that $\by$ still exists in $H_{Z}.$ 
Therefore, we choose $y \in H_{Y}$ such that 
\be
\label{eq:1} 
\by = \argmax_{y \in H_{Y}} \left\{ \mbox{the smallest singular value of }[H_{Z}, \by] \right\}.
\ee 
The chosen $\by$ by (\ref{eq:1}) is the oldest persistent HH. 
Next, we choose the youngest persistent HH $\bx \in H_{X}$ such that  
\be
\label{eq:reason} 
\bx = \argmin_{x \in H_{X}} \left\{ \mbox{the smallest singular value of }[\bx, \by] \right\} = \argmin_{x \in H_{X}} \left\{ 1 - | \bx^{\top} \by |\right\}.
\ee

This procedure is repeated for all persistent holes. 
The incidence matrices are already estimated during the ZC algorithm. 
Since the incidence matrices and their combinatorial Hodge Laplacian are very sparse, the computation of persistent HHs is not so hard in our experiments. 
In our experiments, the total number of persistent holes during the filtration is not more than $50,$ and the number of persistent holes at each filtration value is not more than $20$.

\subsection{HH dissimilarity}

\subsubsection{Bottleneck distance} 

If $K_{a}$ and $K_{b}$ have $m$ and $n$ persistent holes. 
The PDs of $K_{a}$ and $K_{b}$ are denoted respectively by $PD_{a} = \left\{ \bt_{1}^{a}, \cdots, \bt_{m}^{a} \right\}$ and $PD_{b} = \left\{ \bt_{1}^{b}, \cdots, \bt_{n}^{b} \right\},$ where $\bt_{i}$ is a point with the birth and death thresholds of the corresponding hole.
Bottleneck distance between two simplicial complexes, $K_{a}$ and $K_{b}$ is defined by \cite{cohen-steiner.2007.dcg}
$$ D_{B} (K_{a},K_{b}) = d(PD_{a},PD_{b}) = \inf_{\eta: PD_{a} \rightarrow PD_{b}} \sup_{\boldmath{t} \in PD_{a}} \parallel \bt - \eta(\bt) \parallel_{\infty}, $$ where $\eta$ is a bijection from $PD_{a}$ to $PD_{b}$ and $\parallel (x, y) \parallel_{\infty} = \max \left\{ |x|, |y| \right\}$ is the $L_{\infty}-$norm. 
If there is no corresponding hole in the other PD because of $m \neq n,$ 
the points on the diagonal line $x=y$ that have the shortest distance from the point $\bt$ are included.
In this way, the bottleneck distance measures network distance by the difference of the birth and death thresholds of holes, not by the difference between holes themselves.

\subsubsection{Dissimilarity between HHs} 

If the eigenvectors with zero eigenvalues of two different combinatorial Laplacians are denoted by $\bx$ and $\by$, their dissimilarity is defined by one minus the absolute value of their inner product, i.e., 
\be
\label{eq:sim} 
d_{h}(\bx,\by) = 1 - | \bx^{\top} \by |.
\ee 
This is the smallest singular value of the matrix $[\bx, \by]$ in (\ref{eq:reason}) that shows the dependency between $\bx$ and $\by.$ 
If $\bx$ and $\by$ are similar, their dissimilarity is close to 0; otherwise, it is close to 1.

\subsubsection{HH dissimilarity} 

Suppose that two networks $K_{a}$ and $K_{b}$ have $m$ and $n$ persistent HHs, denoted by $\bH_{a} = \left[ \bx_{1}^{a}, \cdots, \bx_{m}^{a} \right]$ and $\bH_{b} = \left[ \bx_{1}^{b}, \cdots, \bx_{n}^{b} \right],$ respectively.
The dissimilarity based on persistent HHs (HH dissimilarity) is defined by 
\be
\label{eq:HH}
D_{H} (K_{a}, K_{b}) = d(\bH_{a},\bH_{b}) = \inf_{\zeta: \boldmath{H}_{a} \rightarrow \boldmath{H}_{b}} \frac{1}{\min(m,n)} \sum_{\bx \in \boldmath{H}_{a}} d_{h}(\bx,\zeta(\bx)), 
\ee
where $\zeta$ is a bijection from $\bH_{a}$ to $\bH_{b}.$  

The correspondence $\zeta$ between persistent HHs in two different networks is determined by minimizing the total distances between the pairs of HHs based on Munkres assignment algorithm, also known as Hungarian algorithm. Some of persistent HHs can not find their corresponding HHs in the other network because of $m \neq n.$ In this study, we ignore them and average the dissimilarities of the obtained pairs of persistent HHs.

\subsubsection{Citation of HH} 
The advantage of using HH dissimilarity is the ability to quantify how much a persistent HH contributes in differentiating 
networks. 
The degree of the contribution of HH is called the citation of HH.
If a persistent HH $\bx$ in $\bH_{a}$ corresponds to a persistent HH $\by = \zeta(\bx)$ in $\bH_{b}$ in (\ref{eq:HH}), their dissimilarity is $d_{h}(\bx,\by) = 1 - | \bx^{\top} \by|,$ and their similarity is defined by $| \bx^{\top} \by|.$
If the persistent HHs of $l$ networks are denoted by $\mathcal{H} = \left\{ \bH_{1}, \cdots, \bH_{l} \right\}$ and they are compared with $\bH_{a},$ the citation of $\bx$ is defined by 
$$\sum_{\zeta(\bx) \in \boldmath{H}, \forall \boldmath{H} \in \mathcal{H}}  | \bx^{\top} \zeta(\bx)|.$$
If we find the most cited HHs by comparing networks within a group, we can determine 
which submodule makes two networks in a group close to each other. 
Furthermore, if we find the most cited HHs by comparing network between groups, we can determine which submodule makes differences. 

\section{Results} 

\subsection{Brain network construction} 

We had 4 groups, NC, sMCI, pMCI, and AD which had 181, 91, 77 and 135 subjects, respectively. 
The subjects in a group could be heterogeneous. 
Thus, we obtained 600 bootstrap samples from each group by randomly selecting the subset of the number of subjects in each group with replacement \cite{sanabria-diaz.2013.plosone}.
The number of bootstrap samples was heuristically determined in comparison with previous study \cite{sanabria-diaz.2013.plosone}.
We constructed 600 bootstrapped networks from bootstrap samples in each group by diffusion distance in Sec. \ref{sec:dataset}. 
The total number of generated brain networks was 2400.
%

\subsection{Network clustering} 

We clustered 2400 bootstrapped brain networks into 4 groups by Ward's hierarchical clustering method. 
The Ward's hierarchical clustering method found the group labels based on the distance between data points, which is a network in our application. 
The network distance was estimated by (a) L2, (b) GH distance, (c) KS$_0,$ (d) KS$_1$, (e) bottleneck distance of holes, and (f) HH dissimilarity \cite{carlsson.2008.ijcv,chung.2017.ipmi,chung.2018.arxiv,cohen-steiner.2007.dcg,lee.2012.ieeemi}.
The obtained distance matrices of 2400 networks were shown in Fig. \ref{fig:networkdistance}. 
After clustering networks, we matched the estimated group label with the true group label of networks and calculated the clustering accuracy of 8 distance matrices.
The clustering accuracy of 8 distance matrices was shown in Table \ref{table:acc}. 
We also clustered 1200 bootstrapped networks in sMCI and pMCI into 2 groups by the same way. 
The clustering accuracy was shown in Table \ref{table:acc}. 

\begin{figure}[t]
\centering
\includegraphics[width=1\linewidth]{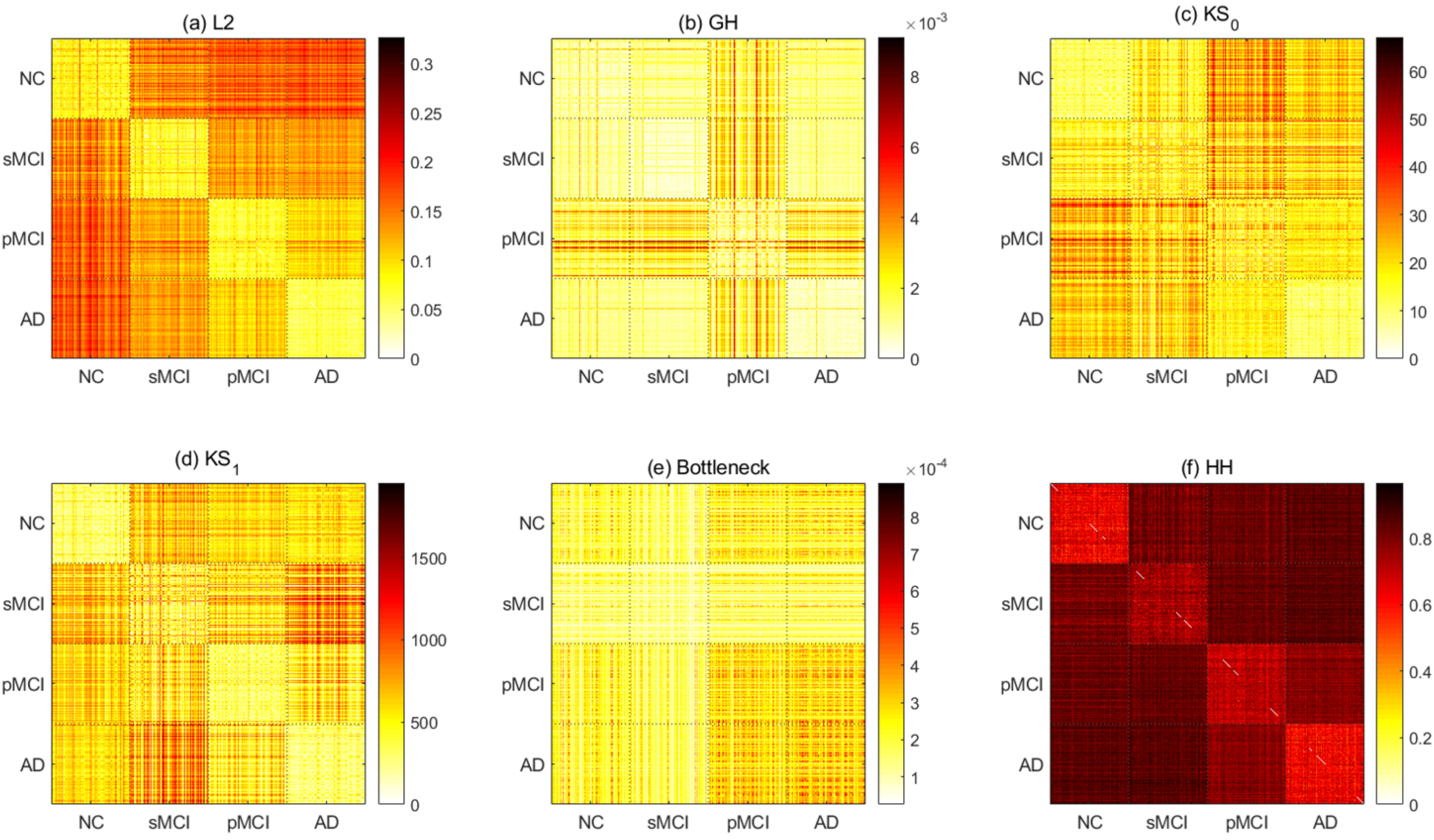}
\vspace*{-0.5cm}
\caption{Distance of 2400 networks. (a) L2, (b) GH, (c) KS$_{0},$ (d) KS$_{1},$ (e) Bottleneck, and (f) HH. The 2400 networks were sorted in the order of NC, sMCI, pMCI, and AD. Each group had 600 networks. The clustering accuracy is shown in Table \ref{table:acc}.}
\label{fig:networkdistance}
\end{figure}

\begin{table}[b]
\caption{Clustering accuracy}  
\center{
\begin{tabular}{| l c | c | c |}
\hline
& \multirow{2}{*}{Distance} & 4 groups  & 2 groups  \\ 
& & (NC, sMCI, pMCI, and AD) & (sMCI and pMCI) \\
\hline 
(a) & L2 & 66.09 \% & 98.50 \% \\
(b) & GH & 45.96 \% & 87.58 \% \\
(c) & KS$_0$ & 52.54 \% & 74.00 \% \\
(d) & KS$_1$ & 77.38 \% &79.83 \% \\
(e) & Bottleneck & 45.71 \% & 76.58 \% \\
(f) & HH & 100 \% &100 \% \\
\hline
\end{tabular}}
\label{table:acc}
\end{table}

\subsection{The most cited HHs} 

We selected the 600 most cited HHs within NC, sMCI, pMCI, and AD, and divided them into 5 clusters based on the dissimilarity between HHs in (\ref{eq:sim}). 
In Fig. \ref{fig:subnetworkswithingroup} (a-d), because the dissimilarity of HHs in the cluster 5 was large, we considered HHs in the cluster 5 as outliers. 
We calculated the center of HHs in clusters 1, 2, 3, and 4, by selecting the HH with the minimum sum of dissimilarities with the other HHs in the cluster. 
The 4 representative HHs of 4 clusters were shown on the left of Fig. \ref{fig:subnetworkswithingroup} (a-d). 
In each panel, the upper row showed the HHs in a brain, and the lower row showed the HHs in a 2-dimensional plane. 
The location of nodes in the 2-dimensional plane was estimated by Kamada-Kawai algorithm implemented in a network analysis/visualization toolbox, Pajek \cite{batagelj.2003.gds}. 
In Fig. \ref{fig:subnetworkswithingroup} (a-d), the width of an edge was proportional to the edge weight in the HH. 
The larger the weight of an edge, the darker the color of an edge.  
The color of nodes represented the location of nodes in a brain. 
If a node was located in frontal, parietal, temporal, occipital, subcortical, and limbic regions, the color of the node was red, blue, green, purple, yellow, and orange, respectively. 

We also selected the 600 most cited HHs when we compared networks between sMCI and pMCI, and divided them into 5 clusters. 
In Fig. \ref{fig:subnetworksbetweengroups} (a), the cluster 5 contained the outliers. 
Thus, we estimated the center HHs in cluster 1-4. 
The representative HHs in sMCI and the corresponding holes in pMCI were shown in Fig. \ref{fig:subnetworksbetweengroups} (b).

\begin{figure}[t]
\centering
\includegraphics[width=1\linewidth]{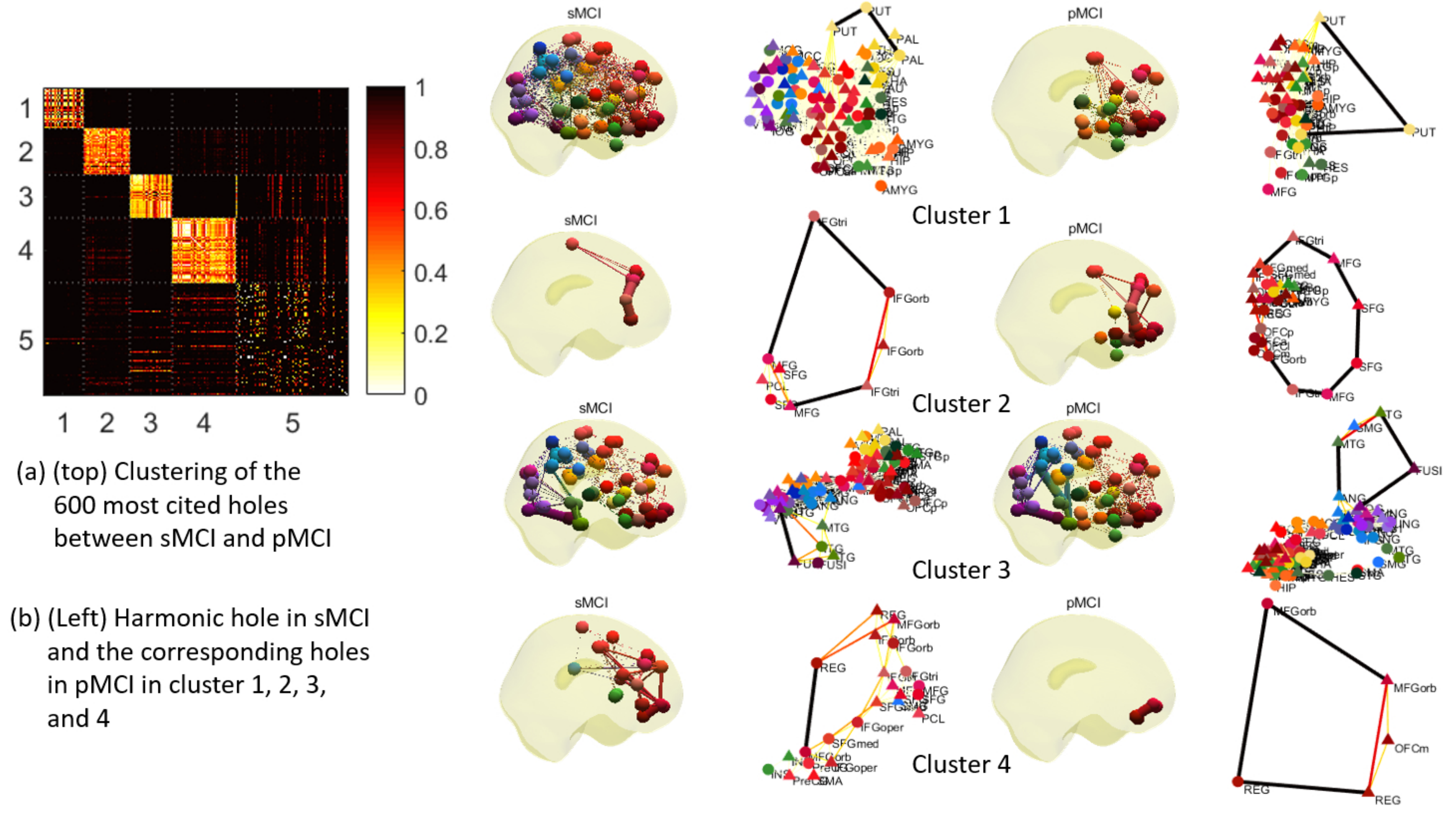}
\vspace*{-0.7cm}
\caption{(a) Clustering of the 600 most cited HHs when sMCI and pMCI were compared. (b) Representative HHs in cluster 1, 2, 3 and 4. The left two columns showed HHs in sMCI and the right two columns showed the corresponding HHs in pMCI. Each HH was visualized in a brain and in a 2-dimensional plane. The shape of the HH was more clearly shown in the plane, and the location of the HH could be checked in the brain. The color of nodes was determined by the location of nodes in a brain:  frontal (red), parietal (blue), temporal (green), occipital (purple), subcortical (yellow), and limbic (orange) regions. If the edge weight was larger in a HH, the color of edge was darker and the width of edge was larger.}
\label{fig:subnetworksbetweengroups}
\end{figure}

\begin{figure}[t]
\centering
\includegraphics[width=1\linewidth]{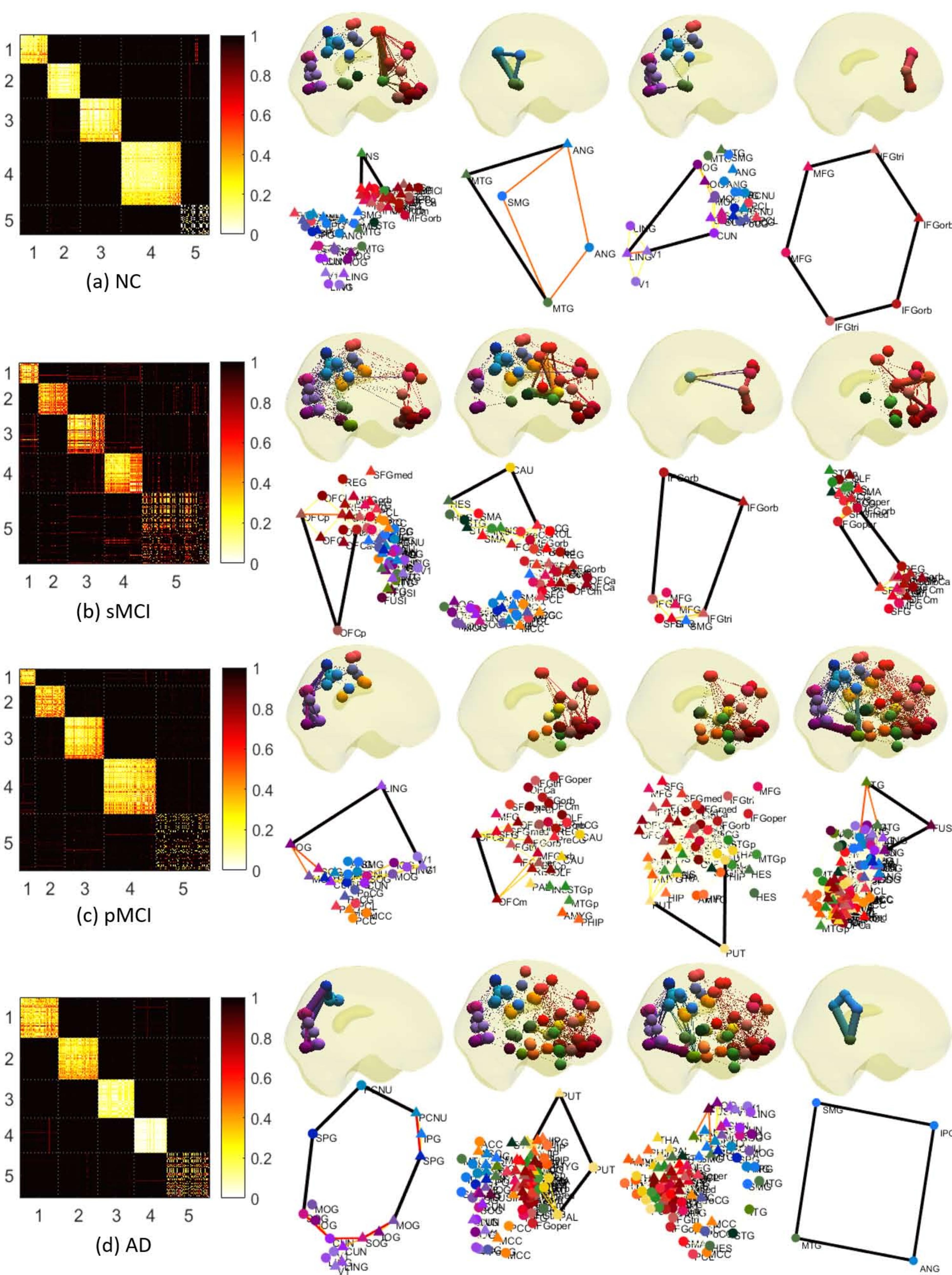}
\caption{Distance matrix of the 600 most cited HHs within (a) NC, (b) sMCI, (c) pMCI, and (d) AD. The most cited holes were clustered into 5 groups. The last cluster 5 had outliers with large dissimilarities between HHs. The representative HHs of the first 4 clusters were plotted on the right. The upper row showed the HHs in a brain and the lower row showed the HHs in a 2-dimensional plane.}
\label{fig:subnetworkswithingroup}
\end{figure}

\section{Discussion and conclusions} 

In this study, we proposed a new network dissimilarity, called HH dissimilarity. 
Unlike a binary hole estimated by the ZC algorithm, a HH show all possible paths of edges around a hole, and the contribution of paths to forming the hole is represented by the weight of edges on the paths. 
If an edge belongs to a unique path that forms a hole, its edge weight will be large.  
If an edge belongs to one of many alternative paths as in a module, its edge weight will be small. 
In this way, HHs can extract the substructures of a brain network including holes and modules. 
Moreover, since the HHs can be represented as 
real-valued orthonormal vectors 
we can define the dissimilarity between HHs as well as HH dissimilarity between brain networks easily using vector product. 

Brain networks of different groups may share common substructure as well as have different substructures that make individual and group differences. 
The proposed HH dissimilarity first finds candidates of common substructures between brain networks and estimates the 
over all dissimilarities between candidates. 
The clustering results showed that brain networks of different groups had similar substructures, however, the averaged similarities was much larger than that of brain networks within a group. 

The goal of persistent homology may be to find persistent features 
that last for a long 
duration. 
However, in brain network analysis, it has been applied for finding the change of topology, especially the change of connected components, 
instead of the persistence of topology.
This study suggested a more coherent framework 
to observe, capture, and quantify the change of holes in brain networks. 
Depending on imaging modality and study populations, brain networks may have different characteristics of shapes. 
Therefore, it is necessary to apply proper network measures to brain networks depending on  modality and population. 
The results showed that when the Alzheimer's disease progresses, the hole structure was changed   in metabolic brain networks, and HHs and HH dissimilarity could predict the disease progression.

\section*{Acknowledgements} 
Data used in preparation of this article were obtained from the Alzheimer's Disease Neuroimaging Initiative (ADNI) database (adni.loni.usc.edu). As such, the investigators within the ADNI contributed to the design and implementation of ADNI and/or provided data but did not participate in analysis or writing of this report. A complete listing of ADNI investigators can be found at  \url{http://adni.loni.usc.edu}. 
This work is supported by Basic Science Research Program through the National Research Foundation (NRF) (No.2013R1A1A2064593 and No.2016R1D1A1B03935463), NRF Grant funded by MSIP of Korea (No.2015M3C7A1028926 and No.2017M3C7A1048079), NRF grant funded by the Korean Government (No. 2016R1D1A1A02937497, No.2017R1A5A1015626, and No.2011-0030815),  and NIH grant EB022856.

\section*{References}

\bibliography{leehk}

\end{document}